\documentclass[11pt,a4paper]{amsart}

\usepackage[
a4paper,
left=32mm,
right=32mm,
top=30mm,
bottom=30mm
]{geometry}

\usepackage[T1]{fontenc}
\usepackage[utf8]{inputenc} 
\usepackage[english]{babel}
\usepackage{microtype}

\usepackage{amsmath}
\usepackage{amssymb}
\usepackage{amsthm}
\usepackage{mathtools}
\usepackage{bm}
\usepackage{mathrsfs}

\usepackage{graphicx}
\usepackage{booktabs}
\usepackage{tikz}
\usepackage{array}

\usepackage{xcolor}
\usepackage[
colorlinks=true,
linkcolor=blue,
citecolor=blue,
urlcolor=blue
]{hyperref}
\usepackage[nameinlink,noabbrev]{cleveref}

\theoremstyle{plain}
\newtheorem{theorem}{Theorem}[section]

\theoremstyle{definition}

\theoremstyle{remark}

\newcolumntype{M}{>{$}c<{$}}

\newcommand{\ZZ}{\mathbb{Z}}

\newcommand{\slalg}{\mathfrak{sl}}
\newcommand{\soalg}{\mathfrak{so}}
\newcommand{\spalg}{\mathfrak{sp}}

\begin{document}
	
	\title{Uniqueness of universal quantum dimensions}
	
	\author{M.Y. Avetisyan and R. L. Mkrtchyan}
	\address{Alikhanyan National Science Laboratory (Yerevan Physics Institute), 2 Alikhanyan Brothers St., Yerevan, 0036, Armenia}
	\email{maneh.avetisyan@gmail.com; mrl55@list.ru}

	\maketitle

	{\small  {\bf Abstract.} 
We prove, under some assumptions, the uniqueness of universal quantum dimension formulae.   
We first show that the quantum dimensions of classical algebras have a specific form that can be  represented universally, in Vogel's sense, as the product/ratio of q-dimension factors with arguments $a\alpha+b\beta+c\gamma$ where  $c=0, 1, 2$ and  $a, b$ are integers. On this basis  we  derive simple uniqueness criteria, according to which all known universal quantum dimensions are unique.

	}

	\tableofcontents

	\section{Introduction}

	For a specific simple Lie algebra, the quantum dimension is defined as the character on the Weyl line in the root space and can be represented as a product over positive roots (see, e.g.,  \cite{FrancescoMathieuSenechal1997}, eq. 13.170): 
	
	\begin{align}\label{qWeyl}
		dim_q(\lambda)=\chi_\lambda (x\rho)= \prod_{\alpha > 0} \frac{\sinh((\alpha,\lambda+\rho)x/2)}{\sinh((\alpha,\rho)x/2)}
	\end{align}
	where $\lambda$ is the highest weight, $\rho$ is the Weyl vector (the sum of fundamental weights), $x=\ln q$ is the parameter of  deformation.  In \eqref{qWeyl}  we assume the minimal normalization of the scalar product in the algebra, i.e., the square of the long root(s) is $2$. 
	The usual dimension can be obtained from (\ref{qWeyl}) in the $x \rightarrow 0$ limit. 
	
	The universal, in Vogel's sense \cite{Vogel1999,Vogel2011}, quantum dimension formulae are the products of the form 
	
	\begin{align}\label{gform}
		\prod_{i=1}^{k} \frac{\sinh(xa_i)}{\sinh(xb_i)}
	\end{align}
	where $a_i, b_i$ are linear functions of Vogel's parameters $\alpha,\beta, \gamma$, e.g. $a_i=x_i\alpha +y_i\beta+z_i\gamma$. 
	
	The main advantage of the universal form (\ref{gform}) is  that, when the parameters are specialized to the values corresponding to a given simple Lie algebra in Vogel's \cref{tab:Vogel}, it gives the quantum dimension of a specific representation of that algebra (with some subtleties related to automorphisms of Dynkin diagrams). These specific representations  belong to the {\it universal multiplet}, i.e. the set of representations whose (quantum) dimensions are obtained from a given universal formula by specializing the Vogel's parameters  to a given algebra, and also from the permutations of that parameters. 
	
	The simplest example of universal quantum dimension is that for adjoint representations \cite{Westbury2003}:
	
	\begin{align}  \label{cad}
		udim_q(\mathfrak{g}) = -\frac{\sinh\left(\frac{(\gamma+2\beta+2\alpha)x}{4}\right)}{\sinh\left(\frac{\gamma x}{4}\right)}\frac{\sinh\left(\frac{(2\gamma+\beta+2\alpha)x}{4}\right)}{\sinh\left(\frac{\beta x}{4}\right)}\frac{\sinh\left(\frac{(2\gamma+2\beta+\alpha)x}{4}\right)}{\sinh\left(\frac{\alpha x}{4}\right)}
	\end{align}
	
	Universal multiplet of this universal quantum dimension formula consists from adjoint representation for each simple algebra, and for the values of parameters from the \cref{tab:Vogel} formula \eqref{cad} gives the quantum dimension of adjoint representation of corresponding algebra. The universal quantum dimension formulae are know almost for all representations, for which the usual universal dimension formulae exist.  
	
	The universal quantum dimensions have a number of applications, particularly in the Chern-Simons theory and the knot theory. An example is the universal expression for the partition function of Chern-Simons theory on the three-dimensional sphere \cite{MkrtchyanVeselov2012,Mkrtchyan2013}:
	
\begin{align}
	-\ln Z &=
	(udim/2)\ln(\delta/t)+
	\int^{\infty}_0 \frac{dx}{x} \frac{f(x/\delta)-f(x/t)}{(e^{x}-1)} \label{Ftotal} \\
	f(x)&=udim_q(\mathfrak{g}), \quad udim=\lim_{q \to 1} udim_q 
\end{align}
	where $\delta$ is shifted coupling of Chern-Simons. For universal knot polynomials see \cite{MironovMkrtchyanMorozov2016,MironovMorozov2016,MironovSingh2026}. 
	
The problem we address in this paper is the problem of the uniqueness of universal quantum dimension formulae.  This problem generally can be formulated as whether there exist another universal formula for a given universal multiplet. Evidently, if we can find the universal formula, called a non-uniqueness factor, of the same kind \eqref{gform} such that it is equal to 1 on all lines from Vogel's \cref{tab:Vogel} and their permutations (totally 12 lines), then we can multiply it by any universal quantum dimension formula and get another formula with the same properties.  

In this form, the problem was studied in \cite{AvetisyanMkrtchyan2022}, where some partial results are obtained. Particularly, the problem was connected  with the theory of configurations of points and lines, and specifically with the existence of the special $(144_{3},36_{12})$ configuration. 
	
In the present paper, we impose weaker requirements on the non-uniqueness factor (thus allowing more general formulae), adapting it to a given universal quantum dimension formula. Namely, we allow the factor to take arbitrary non-zero values on the lines where the given universal quantum dimension vanishes, and require it to be equal to 1 only on the remaining lines. We have found the complete solution of this problem, due to the observation presented below in \cref{thm:1}, concerning the general form of quantum dimensions for a classical algebras, together with a subsequent natural assumption about the general form  of the universal quantum dimension formulae. We then search for a  non-uniqueness factor that preserves that general form, and derive the criteria for its existence.  This is similar to restricting a gauge field to some gauge and asking whether some gauge transformations maintain that gauge. 
	
The answer (\cref{thm:2}) has a very simple form: we assign the   weight 0, 1, or 2 to each of the 12 lines from Vogel's \cref{tab:Vogel} and their permutations, and show that non-uniqueness factor do not exist on the sets of lines with the sum of weights greater than two. This means uniqueness of formulae with non-zero values on those sets of lines. It appears that all existing universal quantum dimension formulae are unique. Particularly, the formula for the universal quantum dimension of the universal multiplet $E$, derived in \cite{AvetisyanMkrtchyan2026},  is unique, as anticipated there.

	\section{The $N$-dependence of quantum dimensions for classical algebras}
	
	Consider quantum dimensions of adjoint representation for classical simple Lie algebras: 
	
	For $\slalg(N)$:
	\begin{eqnarray}\label{qsl}
		\frac{\sinh\left(\frac{x}{2}(N-1)\right) \sinh\left(\frac{x}{2}(N+1)\right)}{\sinh\left(\frac{x}{2}\right)\sinh\left(\frac{x}{2}\right)}
	\end{eqnarray}
	
	For $so(N)$:
	
	\begin{eqnarray}\label{qso}
		\frac{\sinh\left(\frac{x}{4}N\right) \sinh\left(\frac{x}{2}(N-1)\right) \sinh\left(\frac{x}{2}(N-4)\right)}{\sinh\left(\frac{x}{2}\right)\sinh\left(x\right)\sinh\left(\frac{x}{4}(N-4)\right) }
	\end{eqnarray}
	
	For $sp(N)$:
	
	\begin{eqnarray}\label{qsp}
		\frac{\sinh\left(\frac{x}{8}N\right) \sinh\left(\frac{x}{4}(N+1)\right) \sinh\left(\frac{x}{4}(N+4)\right)}{\sinh\left(\frac{x}{4}\right)\sinh\left(\frac{x}{2}\right)\sinh\left(\frac{x}{8}(N+4)\right) }
	\end{eqnarray}

	One can observe that all hyperbolic sines in these formulae have a specific form presented  for any representation in the following theorem:

		\begin{theorem}\label{thm:1}
		Let
		\begin{align} 
		\lambda=\sum_{r=1}^{k}m_r\omega_r,   \quad m_r \in Z_{\geq 0}
		\end{align} 
		be the highest weight of some irreducible representation $V_\lambda$  of $\slalg(N), \soalg(N)$ or $\spalg(N)$, where $\omega_r$ are the fundamental weights,  $k$ and  $m_r$ are fixed independently of sufficiently large $N$. 
		
		Then, it is possible to make some cancellations in the Weyl product \cref{qWeyl} for the quantum dimension  $dim_q(\lambda)$  of the representation $V_\lambda$, so that the only dependence on $N$ will be in the argument of every surviving hyperbolic sine, which will be   the linear  function of $N$ of the form
		\begin{align} 
		\sinh\left(\frac{x}{4}(cN+a)\right), \quad a\in \ZZ
		\end{align} 
		with 
		\begin{align} 
		c\in\{0,2\} \enspace \text{for} \enspace \slalg,  \quad 	c\in\{0,1, 2\} \enspace \text{for} \enspace \soalg,  \quad	c\in\{0, \frac 12, 1 \} \enspace \text{for}\enspace  \spalg
		\end{align} 
	\end{theorem}

	Below we prove the \cref{thm:1}  for the case of $D_n$ algebras, other cases are similar and are considered in Appendixes. 
	
	\begin{proof}

	Consider algebra $so(N)$ with even $N=2n$, so with root system $D_n$. 
	The set of positive roots, and Weyl vector, can be represented as 
	
	\begin{align}
		\Phi^+(D_n)&=\{e_i-e_j, e_i+e_j| 1\leq i <  j \leq n\} \\
		 \rho&=\sum_{i=1}^{n} (n-i)e_i
	\end{align}
	
	Taking into account that $k$ is fixed, and $n$ large, we also have
	
	\begin{align}
		\omega_r=e_1+...+e_r, \quad r\leq k \\
		\lambda=\sum_{i=1}^{k} l_i e_i, \quad l_i=\sum_{r=i}^{k} m_r \quad (i=1,...,k), \quad l_i=0 \quad(i>k)
	\end{align}
	
	The set of positive roots in the Weyl formula \eqref{qWeyl} we divide into three sets. 
	
	{\bf 1. Roots with $i,j \leq k$}
	
	For root $\alpha=e_i-e_j$ we have
	
	\begin{align}
		(\rho,e_i-e_j)=j-i \\
		(\lambda+\rho, e_i-e_j)=j-i+l_i-l_j
	\end{align}
	
	Since both are independent of $n$, they lead to the contribution into \eqref{qWeyl} with $c=0$. 
	
	For root $\alpha=e_i+e_j$ we have
	
		\begin{align}
		(\rho,e_i+e_j)=2n-j-i \\
		(\lambda+\rho, e_i+e_j)=2n-j-i+l_i+l_j
	\end{align}
	
	So, the contribution in  \eqref{qWeyl} will be 
	
	\begin{align}
		\frac{\sinh\left(\frac{x}{4}(2N-2i-2j+2l_i+2l_j)\right) } 	{\sinh\left(\frac{x}{4}(2N-2i-2j)\right)}
	\end{align}
	
	i.e. it corresponds to $c=2$.

	{\bf 2. Roots $e_i-e_j$ with $i\leq k <j$ }

	For root $\alpha=e_i-e_j$ we have
	
	\begin{align}
		(\rho,e_i-e_j)=j-i \\
		(\lambda+\rho, e_i-e_j)=j-i+l_i
	\end{align}

	So, the contribution into  \eqref{qWeyl} will be 
	
	\begin{align}\label{so-ikj}
		\frac{\sinh\left(\frac{x}{4}(2j-2i+2l_i)\right) } 	{\sinh\left(\frac{x}{4}(2j-2i)\right)}
	\end{align}
	
	The contribution of all such roots with fixed $i$ will be the product of \eqref{so-ikj} over $j$ from $j=k+1$ to $j=n$. After telescopic cancellation we have the final answer: 
	
	\begin{align}
		\prod_{r=1}^{l_i}  \frac{\sinh\left(\frac{x}{4}(2n-2i+2r)\right) } 	{\sinh\left(\frac{x}{4}(2k+2r-2i)\right)}
	\end{align}
	
	This is in agreement with our statement that dependence on $N$ is in the arguments, only, and in this case we have  $c=1$. 
	
		{\bf 2. Roots $e_i+e_j$ with $i\leq k <j$ }
	
		For root $\alpha=e_i+e_j$ we have
	
	\begin{align}
		(\rho,e_i+e_j)&=2n-i-j \\
		(\lambda+\rho, e_i+e_j)&=2n-i-j+l_i
	\end{align}

	So, the contribution into  \eqref{qWeyl} will be 
	
	\begin{align}\label{so-ikj-2}
		\frac{\sinh\left(\frac{x}{4}(4n-2i-2j+2l_i)\right) } 	{\sinh\left(\frac{x}{4}(4n-2i-2j)\right)}
	\end{align}
	
	The total contribution of these roots will be the product of \eqref{so-ikj-2} over $j$ from $j=k+1$ to $j=n$. After telescopic cancellation we obtain 
	
	\begin{align}
		\prod_{r=1}^{l_i}  \frac{\sinh\left(\frac{x}{4}(4n-2i-2k-2+2r)\right) } 	{\sinh\left(\frac{x}{4}(2n-2i-2+2r)\right)}
	\end{align}
	
	Again, this expression satisfies all conditions we stated, particularly here $c=1,2$.  
	
	Remaining roots $e_i \pm e_j$ with $i,j >k$ do not contribute since numerator and denominator of \eqref{gform} coincide. 
	
		\end{proof}

	\section{Uniqueness of universal quantum dimension formulae}
	
	One can easily check that the form of quantum dimension of  \cref{thm:1} appears from the expression  (\ref{gform}) with individual hyperbolic sines of the form 
	
	\begin{align}\label{formC}
		\sinh\left(\frac{x}{4}(a\alpha+b\beta+c\gamma)\right)
	\end{align}
	with some integer $a,b$, and $c=0,1,2$.  
	
	This is the abovementioned restricted form of universal quantum dimensions.  When specialized to a given classical algebra, it recovers the form derived in  \cref{thm:1} above. Note particularly  that according to $\spalg$'s parameters  in Vogel's table, this form involves the correct coefficients $c$ for $\spalg$, also.  Actually, we suggest the following
	
	{\bf Conjecture.} {\it All universal quantum dimension formulae have the following form:
	
		\begin{align}\label{gformconj}
		udim_q=\prod_{i=1}^{k} \frac{	\sinh\left(\frac{x}{4}(a_i\alpha+b_i\beta+c_i\gamma)\right)}{	\sinh\left(\frac{x}{4}(a'_i\alpha+b'_i\beta+c'_i\gamma)\right)}
	\end{align}
	where $a_i, b_i, a'_i, b'_i$ are integer, and $c_i, c'_i \in \{0,1,2\}$.}
	
	All known universal quantum dimension formulae have this form.

	Let us have a universal quantum dimension formula $udim_q$ of the form \eqref{gformconj}, which is non-zero on the subset of lines $L=\{l_1, l_2,...\}$ out of 12 lines of simple Lie algebras from Vogel's table, listed in the first column of \eqref{tab:mytable}. Let's call {\it non-uniqueness factor Q of the given universal formula $udim_q$} the universal function $Q$ of the same form \eqref{gformconj}, which is equal to 1 on  all lines from the set $L$, and arbitrary (not singular) on the other lines. The meaning of $Q$ is evident: the product $Q \, udim_q$ is equal to $udim_q$ on all 12 lines of simple Lie algebras and can serve as a universal formula instead of $udim_q$. This definition differs from what we used in \cite{AvetisyanMkrtchyan2022}, since now we allow $Q$ to be  arbitrary on the lines where $udim_q=0$, while earlier we required $Q=1$ on those lines, too. However, now we require $Q$ to maintain the form \eqref{gformconj}. 
	
	So, we are interested in the non-trivial ($Q\neq 1$) non-uniqueness factor $Q$ for a given dimension formula $udim_q$. The following theorem provides an answer to that question. 
	
	\begin{theorem}\label{thm:2}
		Let we have universal quantum dimension formula $udim_q$ in the form \eqref{gformconj}, which is non-zero on the lines $L=\{l_1, l_2,...\}$ and zero on other lines, given in Vogel's \cref{tab:Vogel} and their permutations (see the list in the first column of the \cref{tab:mytable}). Define the weight $w(l_i)$ of each line according to the third column of \cref{tab:mytable} (actually it is the coefficient of $\gamma$ in that line's equation), and define  the weight $w(L)$ of $udim_q$ as a sum  of weights of the lines $l_1, l_2,...$:

		\begin{align}
			w(L)=\sum_{i=1} w(l_i)
		\end{align}
		
		Then, if $w(L) \geq 3$, there is no non-trivial non-uniqueness factor $Q$. 
		  
	\end{theorem}
	
	\begin{proof}
	Let's call hyperbolic sines in the numerator of $Q$ the red sines, and those in the denominator - green ones. The mechanism of transformation of $Q$ into 1 after restriction on some line  is unique: all red factors should cancel with all green ones. 
	
	Let's introduce the generating function

	\begin{align}\label{Pgen}
		P(x,y,u)=\sum_{a,b\in \ZZ} \sum_{c=0}^{2} z_{a,b,c} x^a y^b u^c
	\end{align}
	
	where 
	
	\begin{align}\label{defz}
		z_{a,b,c}=\# red(a,b,c)-\# green(a,b,c)
	\end{align}
	
	where $\# red(a,b,c)$ is the number of red factors in \eqref{gformconj} with universal argument with coefficients $a,b,c$. $\# green(a,b,c)$ is the same for green factors. It is  evident that for given $(a,b,c)$ only one of the numbers $\# red(a,b,c), \# green(a,b,c)$ can be non-zero, otherwise those factors cancel before the restriction of $Q$ on any line. 
	
	Since the total number of red and green factors is finite, the generating function $P(x,y,u)$ has a finite number of terms, despite the infinite range of sums over $a, b$. 
	
	Now let's restrict Vogel's parameters to one of the lines, e.g. $\alpha+\beta=0$, on which $Q=1$.  A linear function in an argument of some arbitrary sinh becomes
	\begin{align}
		a\alpha+b\beta+c\gamma=(a-b)\alpha+c\gamma
	\end{align}
	
	So, all sets $(a,b,c)$ giving the same function differ by an integer multiple of the vector $(1,1,0)$.
	
	We should count the number of functions $p\alpha+q\gamma$ appearing from red factors, and the number of the same functions (i.e. with the  same $p,q$) appearing from green factors. They should be equal, since $Q=1$ on this line. That means, taking into account definition \eqref{defz} of the coefficients of   $P(x,y,u)$, that the sum of all coefficients of $P(x,y,u)$ differing by that  vector in powers of $x,y,u$ has to be zero. Since if we put  $xy=1$ we obtain exactly the sum of coefficients, differing by the vector (1,1,0) in the powers of $x,y,u$, we conclude that  $P(x,y,u)=0$  at $xy=1$, i.e.  because it is the sum of finite number of terms, $P(x,y,u)$ is proportional to $(xy-1)$, i.e. $xy-1$ divides $P(x,y,u)$. 
	
	Similarly, the requirement that $Q=1$ on the line $2\alpha+\beta=0$ leads to the statement that $(x^2y-1)$ is a divisor of  $P(x,y,u)$. 
	Next, the line $\alpha+2\beta=0$ leads to the statement  $(xy^2-1) | P(x,y,u)$. 
	
	We shall call these lines  having  weight 0, according to the power of $u$ in the divisor of  $P(x,y,u)$, which also coincides with the coefficient of $\gamma$ in the equation of the line. This is the case also for all other lines.
	
	The remaining lines are of weights 1 and 2. These, and their corresponding divisors, together with previous zero-weight lines, are given in \cref{tab:mytable}.

\
	 \begin{table}[ht]
	 	\centering
	 	\caption{Lines, divisors, and weights}
	 	\label{tab:mytable}
	 	
	 	\begin{tabular}{|c|c|c|}
	 		\hline
	 		\text{Line} $l$ & \text{Divisor} $q(l)$ & \text{Weight} $w(l)$\\
	 		\hline
	 		$\alpha +\beta$=0 & $xy-1$&0 \\
	 		$2\alpha$+$\beta$=0& $xy^2-1$& 0\\
	 		$\alpha$+2$\beta$=0& $xy^2-1$ & 0\\
	 		$\alpha$+$\gamma$=0&$xu-1$ &1 \\
	 		$\beta$+$\gamma$=0&$yu-1$ & 1\\
	 		2$\alpha$+$\gamma$=0&$x^2u-1$ & 1\\
	 		2$\beta$+$\gamma$=0&$y^2-1$ & 1 \\
	 		$\gamma$=2$\alpha$+2$\beta$&$x^2y^2-u$ & 1 \\
	 		$\alpha$+2$\gamma$=0&$xu^2-1$ & 2\\
	 		$\beta$+2$\gamma$=0&$yu^2-1$ &2 \\
	 		$\beta$=2$\alpha$+2$\gamma$&$x^2u^2-y$ &2 \\
	 		$\alpha$=2$\beta$+2$\gamma$&$y^2u^2-x$ & 2\\
	 		\hline
	 	\end{tabular}
	 	
	 \end{table}
	 
	 No pair of divisors in \cref{tab:mytable}  has a  common divisor. So, from \cref{tab:mytable} one can deduce which sets of lines do not allow a non-uniqueness factor $Q$. Indeed, if $Q$ is equal  to 1 on lines $l_1, l_2...$, then $P(x,y,u)$ is divided by the product of corresponding divisors $q=q(l_1) q(l_2)...$. However, since $P(x,y,u)$ is in maximum quadratic in $u$,  $q$ may be a maximum quadratic in $u$, too. So, if the sum of weights of lines on which $Q$ should be one is $> 2$, i.e. $w(l_1)+w(l_2)+... > 2$,  such $Q$ doesn't exist. 
	 
	 	\end{proof}

	 One can easily list all combinations of lines, which forbids non-uniqueness factors. E.g., if universal quantum dimension formula is non-trivial on all permutations of the exceptional line, it is unique. So,  the formula for the universal multiplet E \cite{AvetisyanMkrtchyan2026},  is unique, as anticipated in \cite{AvetisyanMkrtchyan2026}. 
	 
	 Generally, all known universal quantum dimension formulae \cite{Westbury2003,Mkrtchyan2017,AvetisyanMkrtchyan2020a,AvetisyanMkrtchyan2020,Mkrtchyan2026,AvetisyanMkrtchyan2026} are unique in this sense (to be submitted). 
	 
	 Conversely, we assume that if the the sum of the weights of lines is less than 3, one can always find the solution for $Q$. We didn't try to prove this, instead let's present two examples. For the case of three lines of weight 0, we can present some "random" non-trivial $Q$, which is 1 on those lines:
	 
		\begin{align}
		Q=\frac{\sinh\left(\frac{(-2\beta)x}{4}\right)}{\sinh\left(\frac{\alpha x}{4}\right)}\frac{\sinh\left(\frac{(2\beta+5\alpha)x}{4}\right)}{\sinh\left(\frac{(2\alpha-\beta) x}{4}\right)}\frac{\sinh\left(\frac{(3\beta+4\alpha)x}{4}\right)}{\sinh\left(\frac{(6\alpha+4\beta) x}{4}\right)} 
	   \end{align}
	   
	   We can add one line with weight 1, e.g., the exceptional line $\gamma=2\alpha+2\beta$, so that the sum of weights is 1, and find a solution with k=6. Below,  we present, for that solution, the arguments of the red and green hyperbolic sines without the multiplier $(x/4)$:
	   
	   Red's arguments:
	   \begin{align}
	   	\alpha+\beta, 	6\alpha+5\beta, 	5\alpha+6\beta, 	\alpha+\gamma, 	\beta+\gamma, 	5\alpha+5\beta+\gamma, 
	   \end{align} 
	   Green's arguments:
	     \begin{align}
	   	3\alpha+2\beta, 	2\alpha+3\beta, 	7\alpha+7\beta, 	-\alpha -\beta +\gamma, 	4\alpha+3\beta+\gamma, 	3\alpha+4\beta+\gamma, 
	   \end{align}
	   It is easy to check that on the abovementioned four lines these two multisets coincide.

	   \section{Conclusion}
	   
	   In the above discussion of the uniqueness of universal formulae, we assume that all universal quantum dimension formulae are of the form \eqref{gform}, which is the case for all known formulae. We derive the criteria for the uniqueness of such formulae, and note that it is  fulfilled for all existing ones. It is natural to assume that these features will extend to all other, yet unknown, universal quantum dimension formulae, establishing their uniqueness. We assume that the uniqueness statement can be useful in the search for new universal formulae, as was the case in \cite{AvetisyanMkrtchyan2026}, and in other cases.

	\section*{Acknowledgments}
	We are indebted to Mikayel Mkrtchyan for useful discussions. 
	
	This work was partially supported by the Science Committee of the Ministry of Science and Education of the Republic of Armenia under contracts 21AG-1C060 and 24WS-1C031. 
	
	\section*{AI usage}
	
	ChatGPT was used during the research. Particularly, the proof of \cref{thm:2} was obtained with the assistance of ChatGPT. The other theorems and statements were discovered and proved by standard methods, both manually and with the use of Wolfram Mathematica\texttrademark, and were also checked with ChatGPT. Overall, all proofs and calculations were independently verified by standard methods, and the authors bear full responsibility for the correctness of all results presented in the paper.
	
	\appendix

		\section{Vogel's table}
	
	Here we reproduce  Vogel's table \ref{tab:Vogel} of points in the projective plane corresponding to simple Lie algebras. 
	\begin{table}[h] \caption{Vogel's parameters for simple Lie algebras}     \label{tab:Vogel}
		\begin{tabular}{|r|r|r|r|r|r|} 
			\hline & $\alpha$ &$\beta$  &$\gamma$  & $t=  \alpha+\beta+\gamma$ & Line \\ 
			\hline $sl(N)$ & -2 & 2 & $N$ & $N$ & $\alpha+\beta=0 $\\ 
			\hline $so(N) $ & -2  & 4 & $N-4$ & $N-2$ & $2\alpha+\beta=0$ \\ 
			\hline $sp(2n)$ & -2  & 1 & $n+2$ & $n+1$ & $\alpha+2\beta=0$ \\ 
			\hline $Exc(k)$ & -2 & $k+4$  & $2k+4$ & $3k+6$& $\gamma=2(\alpha+\beta)$ \\ 
			\hline 
		\end{tabular} 
	\end{table}
	
	In the table \ref{tab:Vogel} for the exceptional line $Exc(k)$, $k=-1,-2/3,0,1,2,4,8$ correspond to $A_2,G_2, D_4, F_4, E_6, E_7, E_8$, respectively.

	\section{Quantum dimension of $B_n$ as quantum dimension of $D_n$}

	We consider tensor representations of both algebras, described by a fixed Young diagram
	\begin{align}
	\lambda=(l_1,\ldots,l_r),
	\qquad
	l_1\geq\cdots\geq l_r>0,
\end{align}

	where $r$ and all $l_i$ are independent of $N$ which is much larger than $r$. So, spinor representations are not included.
	
	Set
	\begin{align}
	S(a):=\sinh\left(\frac{x}{2}a\right).
	\end{align}
	
	We shall prove that the Weyl quantum-dimension formula gives the same expression for	$D_n=\mathfrak{so}(2n)$  and  $B_n=\mathfrak{so}(2n+1)$ 	as a function of $N=2n$ in the first case, and $N=2n+1$ in the second case. 
	
	 1. Uniform form of the Weyl vector
	
	For $D_n=\mathfrak{so}(2n)$,
	
	 \begin{align}  \rho_{D_n}=\sum_{i=1}^n(n-i)e_i.  \end{align}
	
	Since $N=2n$, this is
	
	\begin{align}
	\rho_i=\frac {N}{2}-i. \end{align}

	For $B_n=\mathfrak{so}(2n+1)$,
	
	\begin{align}
	\rho_{B_n}=	\sum_{i=1}^n\left(n-i+\frac12\right)e_i.
	\end{align}
		
	Since $N=2n+1$,
	\begin{align}
	n-i+\frac12=\frac N2-i.
	\end{align}
	
	Thus in both cases
	\begin{align}
	\rho_i=\frac N2-i.
	\end{align}
	
	This is the basic reason that a parity-independent formula exists.
	
	The highest weight corresponding to the Young diagram is
	\begin{align} 
	\lambda=\sum_{i=1}^r l_i e_i,
	\end{align} 
	with
	\begin{align}
	l_i=0,\quad i>r.
	\end{align}
	
	2. Roots common to $B_n$ and $D_n$
	
	Both root systems contain
	\begin{align}
	e_i-e_j,\quad e_i+e_j,\quad i<j.
	\end{align}
	
	For $e_i-e_j$,
	\begin{align}
	(\rho,e_i-e_j)=j-i,
	\end{align}
	and
	\begin{align}
	(\lambda+\rho,e_i-e_j)=	j-i+l_i-l_j.
	\end{align}
	
	Therefore the corresponding contribution is
	\begin{align}
	\frac{
		S(j-i+l_i-l_j)
	}{
		S(j-i)
	}.
	\end{align}
	
	For $e_i+e_j$,
	\begin{align}
	(\rho,e_i+e_j)=N-i-j,
	\end{align}
	and
	\begin{align}
	(\lambda+\rho,e_i+e_j)=	N-i-j+l_i+l_j.
	\end{align}
	
	Therefore its contribution is
	\begin{align}
	\frac{
		S(N-i-j+l_i+l_j)
	}{
		S(N-i-j)
	}.
	\end{align}
	
	For $1\leq i<j\leq r$, these contributions are manifestly the same in types $B$ and $D$:
	
	\begin{align}\label{BD1} 
	\prod_{1\leq i<j\leq r}
	\frac{
		S(j-i+l_i-l_j)
	}{
		S(j-i)
	}
	\frac{
		S(N-i-j+l_i+l_j)
	}{
		S(N-i-j)
	}.
\end{align}
	
	If $i,j>r$, then
	\begin{align}
	l_i=l_j=0,
	\end{align}
	so all such factors cancel.
	
	It remains only to compare "tail contributions", i.e. the roots with:
	\begin{align}
	i\leq r<j.
	\end{align}
	
	3. Tail contribution $T_i^{(D)}$ in type $D_n$. 
	
	Now $N=2n$. Consider  $i\leq r$  fixed. 
	
	The roots
	\begin{align}
	e_i-e_j,\quad e_i+e_j,
	\quad j=r+1,\ldots,n,
	\end{align}
	give
	\begin{align}
	T_i^{(D)}=
	\prod_{j=r+1}^{n}
	\frac{S(j-i+l_i)}{S(j-i)}
	\frac{S(N-i-j+l_i)}{S(N-i-j)}.
	\end{align}
	
	The first product telescopes:
	\begin{align}
	\prod_{j=r+1}^{n}
	\frac{S(j-i+l_i)}{S(j-i)}=
	\prod_{p=1}^{l_i}
	\frac{S(n-i+p)}{S(r-i+p)}.
	\end{align}
	
	The second one gives
	\begin{align}
	\prod_{j=r+1}^{n}
	\frac{S(N-i-j+l_i)}{S(N-i-j)}=
	\prod_{p=1}^{l_i}
	\frac{S(N-i-r-1+p)}
	{S(N-i-n-1+p)}.
	\end{align}
	
	Since $N=2n$,
	\begin{align}
	N-i-n-1+p=n-i-1+p.
	\end{align}
	
	Consequently,
	\begin{align}
	T_i^{(D)}=
	\prod_{p=1}^{l_i}
	\frac{
		S(n-i+p)
	}{
		S(n-i+p-1)
	}
	\frac{
		S(N-i-r-1+p)
	}{
		S(r-i+p)
	}.
	\end{align}
	
	The first ratio telescopes once more:
	\begin{align}
	\prod_{p=1}^{l}
	\frac{S(n-i+p)}{S(n-i+p-1)}=
	\frac{S(n-i+l)}{S(n-i)}.
	\end{align}
	
	Using $n=N/2$, we obtain
	
	\begin{align}\label{BD2} 
			T_i^{(D)}=
		\frac{
			S\left(\frac N2-i+l\right)
		}{
			S\left(\frac N2-i\right)
		}
		\prod_{p=1}^{l}
		\frac{
			S(N-i-r-1+p)
		}{
			S(r-i+p)
		}.
	\end{align}
	
	Total contribution is the product of 	$T_i^{(D)}$ over $i$. 
	
	4. Tail contribution $T_i^{(B)}$ in type $B_n$
	
	Now let
	$
	N=2n+1.
	$
		
	The $B_n$ root system contains, in addition to $e_i\pm e_j$, the short positive roots 	$e_i $.

	For fixed $i\leq r$, the short root $e_i$ contributes
	\begin{align}
	\frac{
		S\left(\frac N2-i+l\right)
	}{
		S\left(\frac N2-i\right)
	}.
	\end{align}
	
	The long roots $e_i \pm e_j$ with $j=r+1,...,n$ give
	\begin{align}
	\prod_{j=r+1}^{n}
	\frac{S(j-i+l_i)}{S(j-i)}
	\frac{S(N-i-j+l_i)}{S(N-i-j)}.
	\end{align}
	
	As before, this becomes
	\begin{align}
	\prod_{p=1}^{l_i}
	\frac{S(n-i+p)}{S(r-i+p)}
	\frac{
		S(N-i-r-1+p)
	}{
		S(N-i-n-1+p)
	}.
	\end{align}
	
	But now
	$
	N=2n+1,
	$
	so 	$	N-i-n-1+p=n-i+p.	$
	
	Therefore the factors
	$
	S(n-i+p)
	$
	cancel exactly, leaving
	\begin{align}
	\prod_{p=1}^{l_i}
	\frac{
		S(N-i-r-1+p)
	}{
		S(r-i+p)
	}.
	\end{align}
	
	Multiplying by the short-root contribution we obtain
	\begin{align}\label{BD3}
		T_i^{(B)}=
		\frac{
			S\left(\frac N2-i+l_i\right)
		}{
			S\left(\frac N2-i\right)
		}
		\prod_{p=1}^{l_i}
		\frac{
			S(N-i-r-1+p)
		}{
			S(r-i+p)
		}.
	\end{align}
	
	This is the same function as that for $D_n$ \eqref{BD2}, simply taken for odd $N$. 
	
	5. Uniform formula
	
	Combining the finite part and the tail part, we obtain the final formula for quantum dimensions of $B_n$  and $D_n$ series, joining them into one function for $\soalg(N)$ algebra:

		\begin{align}
			\dim_q^{\mathfrak{so}(N)}V_\lambda
			=&
			\prod_{1\leq i<j\leq r}
			\frac{
				S(j-i+l_i-l_j)
			}{
				S(j-i)
			}
			\frac{
				S(N-i-j+l_i+l_j)
			}{
				S(N-i-j)
			}
			\\
			&\times
			\prod_{i=1}^{r}
			\frac{
				S\left(\frac N2-i+l_i\right)
			}{
				S\left(\frac N2-i\right)
			}
			\prod_{p=1}^{l_i}
			\frac{
				S(N-i-r-1+p)
			}{
				S(r-i+p)
			}.
		\end{align}
	where 
	\begin{align}
	S(a)=\sinh\left(\frac{x}{2}a\right).
	\end{align}
	
	Thus the $B_n$ and $D_n$ expressions are restrictions, respectively to odd and even integer values of $N$, of the same universal $\soalg(N)$ quantum-dimension function. Since we already proved \cref{thm:1} for $D_n$, this establishes that for $B_n$.

	\section{Theorem \eqref{thm:1} for $\slalg$ and $\spalg$}

	\begin{proof}
		
		For $\slalg(N) \sim A_{N-1}$,  in the same orthonormal basis $e_i$, the positive roots, fundamental weights and Weyl vector are
		
		\begin{align}
		\Phi^+(A_{N-1})&=\{e_i-e_j| 1\leq i <  j \leq N\} \\
		\omega_i &= e_1+...+e_i - \frac{i}{N} \sum_{j=1}^{N} e_j \\
		\rho&= \frac{1}{2} \sum_{i=1}^{N} (N+1-2i)e_i
		\end{align}

		Coefficients of the decomposition of $\lambda$ over $e_i$ are 
		\begin{align}
		l_i=\sum_{r=i}^{k}m_r,\quad i\le k,
		\end{align}
		and  $l_i=0$ for $i>k$. 
		We have 
		\begin{align}
		(\rho,e_i-e_j)=j-i,\qquad
		(\lambda,e_i-e_j)=l_i-l_j.
		\end{align}
		Hence the Weyl formula becomes
		\begin{align}
		\dim_qV_\lambda=
		\prod_{1\le i<j\le N}
		\frac{
			\sinh\left(\frac{x}{4}\,2(j-i+l_i-l_j)\right)
		}{
			\sinh \left(\frac{x}{4}\,2(j-i)\right)
		}.
		\end{align}
		
		We split the product into three parts.
		
		If $i,j\le k$, all arguments are independent of $N$, so these factors contribute only terms with $c=0$.
		
		If $i,j>k$, then $l_i=l_j=0$, and the corresponding numerator and denominator cancel, contribution is 1.
		
		Finally, let $i\le k<j$. For fixed $i$, the corresponding contribution is
		\begin{align}
		\prod_{j=k+1}^{N}
		\frac{
			\sinh\left(\frac{x}{4}\,2(j-i+l_i)\right)
		}{
			\sinh\left(\frac{x}{4}\,2(j-i)\right)
		}.
		\end{align}
		This product telescopes to
		\begin{align}
		\prod_{p=1}^{l_i}
		\frac{
			\sinh\left(\frac{x}{4}\,2(N-i+p)\right)
		}{
			\sinh \left(\frac{x}{4}\,2(k-i+p)\right)
		}.
		\end{align}
		
		Therefore every surviving hyperbolic sine has argument of the form
		\begin{align}
		\frac{x}{4}(cN+a),
		\quad c\in\{0,2\}.
		\end{align}
		
	with integer $a$. This finishes the proof. 
	
	Actually, we need the same statement for the representations with Dynkin labels of type 
	
	\begin{align}(\lambda_1,..., \lambda_k,...,\tau_k,...,\tau_1)\end{align}
	
	with $N$-independent $k,\lambda_i, \tau_i$. However, the proof is similar, and we omit it. 
	\end{proof}

	Next we prove the theorem for $\spalg(N)$ algebra.

	\begin{proof}
		
		Now we choose normalization of orthogonal basis $e_i$ as
		
		\begin{align}
		(e_i,e_j)=\frac{1}{2}\delta_{ij},
		\end{align}
		so that the long roots have squared length $2$, as required. 
		
		Again denote 
		\begin{align}
		l_i=\sum_{r=i}^{k}m_r,\quad i\leq k,
		\end{align}
		and $l_i=0$ for $i>k$. The positive roots of $C_n$ are
		\begin{align}
		e_i-e_j,\quad e_i+e_j\quad (i<j),\quad 2e_i.
		\end{align}

		In the chosen normalization, their contributions to the Weyl product are respectively 
		\begin{align}
		\frac{
			\sinh\left(\frac{x}{4}(j-i+l_i-l_j)\right)
		}{
			\sinh \left(\frac{x}{4}(j-i)\right)
		},
		\end{align}
		\begin{align}
		\frac{
			\sinh\left(\frac{x}{4}(N-i-j+2+l_i+l_j)\right)
		}{
			\sinh\left(\frac{x}{4}(N-i-j+2)\right)
		},
		\end{align}
		and
		\begin{align}
		\frac{
			\sinh\left(\frac{x}{4}(N-2i+2+2l_i)\right)
		}{
			\sinh\left(\frac{x}{4}(N-2i+2)\right)
		},
		\end{align}
		respectively.
		
		The factors with $i,j>k$ cancel identically. For fixed $i\leq k$, the
		contribution of roots $e_i - e_j$ telescopes as
		\begin{align}
		\prod_{j=k+1}^{n}
		\frac{
			\sinh\left(\frac{x}{4}(j-i+l_i)\right)
		}{
			\sinh\left(\frac{x}{4}(j-i)\right)
		}
		=
		\prod_{p=1}^{l_i}
		\frac{
			\sinh\left(\frac{x}{4}(n-i+p)\right)
		}{
			\sinh\left(\frac{x}{4}(k-i+p)\right)
		}.
		\end{align}
		Since $n=N/2$, these sines have    the form 	$\frac{x}{4}(cN+a)$ with      coefficients $c=1/2$ and $c=0$, and integer $a$.
		
		Similarly, the 	contribution of roots $e_i + e_j$ telescopes as
		\begin{align}
		\prod_{j=k+1}^{n}
		\frac{
			\sinh\left(\frac{x}{4}(N-i-j+2+l_i)\right)
		}{
			\sinh\left(\frac{x}{4}(N-i-j+2)\right)
		}
		=
		\prod_{p=1}^{l_i}
		\frac{
			\sinh\left(\frac{x}{4}(N-i-k+1+p)\right)
		}{
			\sinh\left(\frac{x}{4}(n-i+1+p)\right)
		}.
		\end{align}
		Its arguments therefore have coefficients $c=1$ and $c=1/2$.
		
		The remaining roots have indices at most $k$. 	Thus every surviving argument is of the form
		\begin{align}
		a,\quad \frac{N}{2}+a,\quad N+a,
		\end{align}
		with integer $a$, and consequently
		\begin{align}
		c\in\left\{0,\frac12,1\right\}.
		\end{align}
	as stated.
	\end{proof}

	\end{document}